\documentclass[a4paper]{article}

\usepackage[T1]{fontenc}
\usepackage{hyperref}
\usepackage{amsmath}
\usepackage{amssymb}
\usepackage{amsthm}
\usepackage{stmaryrd}
\usepackage{mathtools}

\newcommand{\M}[1]{\ensuremath{\mathcal{#1}}}
\newcommand{\mucalc}{\ensuremath{L_{\mu}}}
\newcommand{\nxt}{\mathop{\bigcirc}}
\newcommand{\fp}{\mathrm{fp}}
\newcommand{\sem}[3]{\llbracket #1 \rrbracket_{#3}^{#2}}
\newcommand{\factor}[1]{\ensuremath{F(#1)}}
\newcommand{\ifactor}[1]{\ensuremath{F_\infty(#1)}}
\newcommand{\recurrence}[2]{\ensuremath{R_{#2}(#1)}}

\newcommand{\ajoin}[4]{\ensuremath{\mathrm{join}(#1, #2, #3, #4)}}
\newcommand{\aremove}[1]{\ensuremath{\mathrm{remove}(#1)}}
\newcommand{\aout}[1]{\ensuremath{\mathrm{out}(#1)}}
\newcommand{\laout}[1]{\ensuremath{L_\mathrm{out}(#1)}}

\newcommand{\cl}[2]{\ensuremath{\mathrm{Cl}(#1, #2)}}
\newcommand{\myfalse}{\texttt{ff}}
\newcommand{\mytrue}{\texttt{tt}}

\def\newarrow#1{\mathop{{\hbox{\setbox0=\hbox{$\scriptstyle{#1\quad}$}{$%
\mathrel{\mathop{\setbox1=\hbox to
\wd0{\rightarrowfill}\ht1=3pt\dp1=-2pt\box1}\limits^{#1}}%
$}}}}}
\newcommand{\Transition}[3]{\ensuremath{#1 \newarrow{#2} #3}}

\theoremstyle{plain}
\newtheorem{theorem}{Theorem}
\newtheorem{lemma}[theorem]{Lemma}
\newtheorem{proposition}[theorem]{Proposition}

\theoremstyle{definition}
\newtheorem{definition}[theorem]{Definition}
\newtheorem{remark}[theorem]{Remark}

\allowdisplaybreaks

\begin{document}
\title{Finite Convergence of the Modal $\mu$-Calculus on Almost-Periodic Words}

\author{Fabian Lehr \\ TU Munich, Germany \and Florian Bruse \\ TU Munich, Germany}
\maketitle             
\begin{abstract}
A formula of the modal $\mu$-calculus enjoys finite convergence on a structure if there is
some finite unfolding of the formula that defines the same set. A structure enjoys finite
convergence if all formulas of the $\mu$-calculus enjoy finite convergence on said structure. It is
known that there are words that are not ultimately periodic, but have finite convergence.

An almost-periodic word $w$ is one in which each finite word $v$ either appears only finitely often,
or within each factor of some length that only depends only on $w$ and $v$. It is immediate
that words that have finite convergence must be almost periodic. In this paper we show the converse,
namely that all almost-periodic words have finite convergence.
This characterizes finite convergence on infinite words, and also re-proves a decidability result due to Semenov ('84).

\end{abstract}
 
\section{Introduction}
\label{sec:intro}
The Modal $\mu$-calclus ($\mucalc$) \cite{DBLP:conf/focs/Pratt81,DBLP:journals/tcs/Kozen83}
is a central logic for model-checking and verification, in particular due to its role
as the yardstick for regular temporal logics such as LTL and CTL. $\mucalc$ extends
basic modal logic by the ability to formulate recursive definitions via least and greatest
fixpoints. 

Since $\mucalc$ serves as the prototypical temporal logic, the exact nature
of its recursive definitions has been studied intensively. For example, it is
known that nesting more fixpoint definitions of opposite polarity (e.g., least fixpoints that
are mutually recursive with greatest fixpoints, etc.) gives strictly more expressive power \cite{DBLP:conf/concur/Bradfield96}. Hence, the so-called \emph{alternation hierarchy} of the $\mu$-calculus is strict.
This is not true if one restricts oneself to words or word-like structures \cite{DBLP:conf/concur/Kaivola95,DBLP:journals/tcs/GutierrezKL14},
where the $\mu$-calculus collapses into its alternation-free fragment.
Finally, over structures without infinite paths, least fixpoints can be rewritten into greatest fixpoints and vice versa \cite{DBLP:conf/tacas/Mateescu02}.

Another central question is how fast the fixpoint of such a recursive definition is reached. It follows from the Kleene
Fixpoint Theorem \cite{Kleene:1938} that both least and greatest fixpoints can be replaced by an ordinal-indexed sequence of 
approximations, such that, over any given structure, some approximation is equivalent to the actual fixpoint definition.
Hence, the fixpoint definition \emph{converges} at this ordinal.
This ordinal is then called the structure's closure ordinal of the fixpoint definition, respectively that of the formula that expresses it. 
The exact nature of the closure ordinal has consequences for e.g., model checking \cite{DBLP:reference/mc/BradfieldW18}.

Given a formula, the least ordinal that bounds the closure ordinal of the formula on \emph{all} structures is simply called
the closure ordinal of the formula, if it exists. For example, the formula $\mu X.\ \Box x$, which expresses that all
paths in the structure are finite, has no closure ordinal.
 There has been a long chain of research establishing the exact nature of closure ordinals
of formulas. For example, \cite{DBLP:conf/fics/Czarnecki10} establishes that for every ordinal less that $\omega^2$, 
there is a formula with this closure ordinal. \cite{DBLP:journals/lmcs/GouveiaS19} gives formulas with uncountable closure ordinals, and \cite{DBLP:conf/csl/AfshariL13,DBLP:journals/corr/abs-2511-02594}
show for the so-called $\Sigma$-fragment that closure ordinals are either below $\omega^2$ or at least $\omega_1$.
Moreover, \cite{DBLP:conf/stacs/NiwinskiPS25} presents a similar dichotomy for monadic second-order logic on trees.
See \cite{DBLP:journals/corr/abs-2511-02594} for a good overview over the literature.

A related line of research concerns the closure behavior of \emph{structures}, i.e., the closure ordinals of formulas on a fixed structure. On finite structures, all fixpoint definitions uniformly
converge at the cardinality of the structure, and this extends to structures with finite bisimulation quotient,
where convergence happens at the size of the quotient. A natural question is then whether finite convergence 
also requires a finite bisimulation quotient. This is of interest because e.g.,
on structures on which all fixpoint definitions are equivalent to \emph{some} finite approximation, this approximation
process can be internalized into a stronger logic, yielding a collapse of the alternation hierarchy of $\mucalc$ into
the alternation-free fragment of this logic \cite{BruseLL18}.

 This question was resolved
in \cite{DBLP:conf/mfcs/BruseSL21} by exhibiting a word with infinite bisimulation quotient, on which nevertheless
all formulas of the $\mu$-calculus converge after a finite number of approximations. This number depends on the formula,
whence the closure ordinal of this word is still $\omega$, but it is relatively easy to see that that only structures
with finite bisimulation quotient can have finite closure ordinals (cf.\ Lemma~\ref{lem:closure-trivial} below). What was
left open in \cite{DBLP:conf/mfcs/BruseSL21} is the exact nature of the class of words on which all fixpoint definitions
have finite closure ordinal, i.e., the class of words with \emph{finite convergence}.
This paper settles this question by showing that a word has finite convergence 
iff it is almost periodic. An almost-periodic word $w$ is one where each factor $v$ either appears only finitely often,
or within each factor of length $k_v$, where $k_v$ depends on $v$ and $w$. Intuitively, $v$ re-appears in $w$ after
at most $k_v$ steps, but not necessarily in exactly $k_v$ steps. 

This characterization of the words with finite convergence also yields a decidability
result for this class, provided that the word be effectively constructible and the $k_v$ are computable from $v$ and $w$.
This was already shown by Semenov \cite{semenov}. A central step in Semenov's result and our proof are very similar, as this
concerns the behavior of deterministic automata on almost-periodic words (Lemma~\ref{lem:fac:finite} below, Thm.~3 in e.g., \cite{DBLP:journals/tcs/MuchnikSU03}).
In fact, our result can be obtained directly from some formulations of this result (cf.\ Thm.~3.1.3 in \cite{karimov2023algorithmic}).
We opt to present a direct proof in this paper, for two reasons: first, the proof in e.g., \cite{DBLP:journals/tcs/MuchnikSU03}
is not strong enough to yield our result directly, while that in \cite{karimov2023algorithmic} is rather unwieldy. Our proof,
on the other hand, is relatively straightforward and uses very basic tools. Second, our direct presentation also
yields a normal-form result for omega-regular expressions over almost-periodic words (see Rem.~\ref{rem:nf} below). 
While related results exist \cite{DBLP:journals/pacmpl/AlmagorKKO021}, this version is new as far as we know.

The paper is structured as follows. In Sect.~\ref{sec:prel} we recall the linear-time $\mu$-calculus, some basic automata theory,
and some basic definitions on infinite words. In Sect.~\ref{sec:conv} we present definitions around convergence of fixpoint
definitions, recall some known results, and show some rather simpler characterizations. Sect.~\ref{sec:finite} contains
the proof that all fixpoints are equivalent to a finite approximation on almost-periodic words. We conclude in Sect.~\ref{sec:concl}.

\section{Preliminaries}
\label{sec:prel}
\paragraph*{Words and Languages.}

An \emph{alphabet} $\Sigma$ is a finite set. A $\Sigma$-\emph{word} is a finite or infinite sequence of symbols $w = a_0\dotsb a_{n-1}$ or $a_0 \dotsb$; its \emph{length} $|w|$ is $n$ in the former case, and $\omega$
in the latter; $\epsilon$ is the empty word. For $w = a_0\dotsb$ of length $> i$, we let $w(i)$ denote $a_i$, and $w^i$ the suffix
of $w$ starting with $w(i)$, inclusively. \emph{Concatenation}
of words is indicated by juxtaposition; note that concatenation of $u$ and $v$ implies
that $u$ is a finite word. For a word $w$, let $\factor{w}$ denote the set of \emph{factors} of $w$, i.e.,
the set of $v$ such that $w = uvv'$, and for an infinite $w$, let $\ifactor{w}$ denote the set of
factors of $w$ that occur infinitely often in $w$.

An infinite word is \emph{periodic} if it is of the form $u v^\omega$ for some finite, nonempty $u,v$. 
A factor $u$ is \emph{recurrent} in $w$ if it occurs infinitely often in $w$; it is \emph{uniformly recurrent}
if it occurs in every factor of length $k$ for some $k$ depending on $u$ and $w$. In this case, the smallest
such $k_u$ is called the \emph{recurrence distance} of $u$, we denote it by $\recurrence{u}{w}$, or just $\recurrence{u}{}$ if $w$ is clear from context. An infinite word $w$ is called \emph{uniformly recurrent} if every factor either is uniformly recurrent in $w$, or it does not appear in $w$ at all.
An infinite word $w$ is \emph{almost periodic} if
every factor $u$ is either uniformly recurrent in $w$ or it does not appear in $w$ after the prefix of some length $k_u$. In this case, we also refer to $k_u$ as the recurrence distance.

An infinite word $w$ is \emph{effectively constructible} if there is a computable function $f_w\colon \mathbb{N} \to \Sigma$ that maps $n$ to $w(n)$.
 It is \emph{effectively almost periodic}, resp.\ \emph{effectively uniformly recurrent},
if it is effectively constructible and there is a computable function $f \colon \Sigma^* \to \mathbb{N}$ that
returns the recurrence distance for any finite word.

$\Sigma^*$ is the set of all finite $\Sigma$-words, and $\Sigma^\omega$ is the set of
all infinite $\Sigma$-words. 
A \emph{language} of finite or infinite words is a subset of $\Sigma^*$, resp.\ $\Sigma^\omega$. 
Concatenation $L_1 \cdot L_2$ is defined as $\{uv \mid u \in L_1, v \in L_2\}$, 
and for a language of finite words $L$, the \emph{iteration} $L^*$ is defined as $\bigcup_{i \geq 0} L^i$
where $L^0 = \{\epsilon\}$ and $L^{i+1} = L \cdot L^i$.

A language $L$ of finite words is \emph{prefix-free} if there are no $u,v \in L$ such 
that $v = v_1 v_2$ with $v_1 = u$ and $v_2 \not = \epsilon$. Note that prefix-free
languages cannot simultaneously contain $\epsilon$ and another word.

\paragraph*{Regular Expressions and Automata.}

Regular expressions (for finite $\Sigma$-words) are defined inductively from the grammar
\[
r \Coloneqq \emptyset \mid  \epsilon \mid a \mid r + r \mid r \cdot r \mid r^*
\]
where $a \in \Sigma$. We write $\text{REG}(\Sigma)$ for the set of regular expression over $\Sigma$. We may use typical syntactic sugar such as complementation $\overline{r}$.
A regular expression $r$ defines a finite-word language $L(r)$, inductively given via
\begin{align*}
L(\emptyset) &= \emptyset & L(\epsilon) &= \{\epsilon\} \\
L(a) &= \{a\} & L(r_1 + r_2) &= L(r_1) \cup L(r_2) \\
L(r_1 \cdot r_2) &= L(r_1) \cdot L(r_2) & L(r^*) &= (L(r))^*.
\end{align*}

A \emph{Nondeterministic Finite Automaton with Regular Expressions} (NFA$_\text{reg}$) is an $\M{A} = (Q, \Sigma, \delta, I, F)$ 
where $Q$ is a finite, nonempty set of states, $\Sigma$ is an alphabet, $I,F \subseteq Q$
are the sets of \emph{initial} resp.\ \emph{final} states, and $\delta \subseteq Q \times \text{REG}(\Sigma) \times Q$ is the \emph{transition relation}, which is assumed to be finite. A \emph{run} of $\M{A}$ on a word $w = a_0 \dotsb a_{n-1}$ is partition $w = u_0 \dotsb u_{m-1}$
of $w$ together with a sequence $q_0,\dotsc,q_m$ of states such that (\textsc{i}) $q_0 \in I$
and (\textsc{ii}) for all $0 \leq i < m$ there is $(q_i, r_i, q_{i+1}) \in \delta$ such that
$u_i \in L(r_i)$. A run is \emph{accepting} if it ends in a final state. The language $L(\M{A})$ of $\M{A}$
is the set of all finite words on which there is an accepting run of $\M{A}$. 
Given $q \in Q$, define $L(\M{A}_q)$ as $L(\M{A}')$ where $\M{A}'$ is obtained from $\M{A}$ by making
$q$ the unique initial state.

A special case of an NFA$_\text{reg}$ is the ordinary NFA, in which 
all regular expressions are of the form $a$, i.e., all transitions are labeled by, potentially
different, single-word languages that each contain a word of length $1$. Such an NFA
is called \emph{deterministic} or DFA if, for each $q \in Q, a \in \Sigma$, there is at most one $q' \in Q$
such that $(q, a, q') \in \delta$, and also $|I| = 1$. It is well-known that all these models
are equivalent and define exactly the regular languages.
See e.g.\ \cite{esparza2023automata} for a more thorough exposition of NFA$_\text{reg}$.
 
A \emph{Trivial Automaton} (TrA) \cite{DBLP:conf/mfcs/BruseSL21} is an $\M{A} = (Q, \Sigma, \delta, q_I, F, b)$ where
$Q, \Sigma, F$ are as for NFA$_\text{reg}$, $\delta \subseteq (Q\setminus F) \times \Sigma \times Q$
is deterministic on $Q \setminus F$ and $q_I$ is the unique initial state while $b$ is a bit. Hence,
such a TrA is almost a DFA, except that it works on infinite words and that final states have no successors.
Hence, a run of $\M{A}$ on  $w = a_0\dotsb$ is a finite or infinite sequence $q_I = q_0,\dotsc$
such that either (\textsc{i}) for all $i \geq 0$ we have that $(q_i, a_i , q_{i+1}) \in \delta$,
or such that there is $j \geq 0$ such that, for all $0 \leq i < j$ we have that $(q_i, a_i,q_{i+1}) \in \delta$
and $q_j \in F$. Such a run is accepting if either $b =  0$ and the run is infinite, or if $b=1$ and the run is finite (and, hence, ends in a final state). 
The intuition here is that a TrA is a parity automaton with only one  priority, indicated by $b$, and that runs can additionally
stop early on reaching a final state in which the automaton either accepts or rejects, in a way opposite to that
indicated by its lone priority.

A trivial automaton $\M{A}$ has \emph{bounded runs} on some word $w$ if all runs of $\M{A}$ either stop after at most $k$ steps for some $k$ depending
only on $w$ and $\M{A}$, or go on forever. A word $w$ has bounded runs if all trivial automata have bounded runs on $w$ (but not necessarily
with the same bound).

\paragraph*{The Linear-Time $\mu$-Calculus.}

Let $\Sigma$ be an alphabet, and let $\M{X}$ be a set of \emph{variables}. The syntax of the \emph{Linear-Time $\mu$-calculus} (\mucalc) is given inductively by the following grammar
\[
\varphi \Coloneqq a \mid \varphi \vee \varphi \mid \neg \varphi \mid \nxt \varphi \mid X \mid \mu X.\ \varphi
\]
where $a \in \Sigma$ and $X \in \M{X}$. Moreover, as typical for the $\mu$-calculus, all occurrences of $X$
in $\mu X.\ \varphi$ must occur under an even number of negations.
The notions of free and bound variables are defined as per usual, and we assume that each variable $X$ is bound
at most once in a formula $\varphi$, giving rise to a function $\fp_\varphi$ where $\fp_\varphi(X)$ is the unique formula $\psi$
where $\mu X.\ \psi$ is a subformula of $\varphi$.

We use syntactic sugar such as $\wedge$ and $\nu X.\ \varphi$, where the latter defines greatest fixpoints
as opposed to the least fixpoints defined by $\mu X.\ \varphi$. 

The semantics of an $\mucalc$ formula $\varphi$ on an infinite word $w = a_0 \dotsb$ makes use of \emph{variable assignments}; such a variable assignment is an $\alpha \colon \M{X} \to \mathbb{N}$. Its update $\alpha[X \mapsto S]$ for some $S \subseteq \mathbb{N}$ is defined as usual.
The semantics $\sem{\varphi}{w}{\alpha}$ of $\varphi$ on $w$, relative to $\alpha$ is then
a subset of $\mathbb{N}$ defined 
inductively via
\begin{align*}
\sem{a}{w}{\alpha} &= \{i \in \mathbb{N} \mid w(i) = a\} \\
\sem{\psi_1 \vee \psi_2}{w}{\alpha} &= \sem{\psi_1}{w}{\alpha} \cup \sem{\psi_2}{w}{\alpha} \\
\sem{\neg \psi}{w}{\alpha} &= \mathbb{N} \setminus \sem{\psi}{w}{\alpha} \\
\sem{\nxt \psi}{w}{\alpha} &= \{i \in \mathbb{N} \mid i+1 \in \sem{\psi}{w}{\alpha}\} \\
\sem{X}{w}{\alpha} &= \alpha(X) \\
\sem{\mu X.\ \psi}{w}{\alpha} &= \bigcap \{S \subseteq \mathbb{N} \mid \sem{\psi}{w}{\alpha[X \mapsto S]} \subseteq S\}.
\end{align*}
The least fixpoint in the last line is guaranteed to exist since all the operators (with the exception of negation, which is subject to the restriction mentioned above) have monotone semantics. It is clear that 
we can drop $\alpha$ for closed formulas, whence we also write $w(i) \models \varphi$ to indicate
$i \in \sem{\varphi}{w}{\alpha}$ for an $\alpha$. We might also write $w \models \varphi$ to indicate
that $w(0) \models \varphi$, and similarly for e.g., $w^i$.

Two $\mucalc$ formulas $\varphi_1, \varphi_2$ are \emph{equivalent}, written $\varphi_1 \equiv \varphi_2$, if $\sem{\varphi_1}{w}{\alpha} = \sem{\varphi_2}{w}{\alpha}$ for all $w, \alpha$.
Since $\sem{\mu X.\ \psi}{w}{\alpha}$ is a fixpoint, it is immediate that $\mu X.\ \psi \equiv \psi[\mu X.\ \psi/X]$ where the latter indicates syntactic substitution of $X$ by $\mu X.\ \psi$.

\section{Fixpoint Convergence}
\label{sec:conv}

Let $\M{T}$ be a some transition system (e.g., a word). It is well known that any $\mu$-calculus formula of the form $\mu X.\ \varphi$ defines an ordinal-indexed sequence of approximations
via $X^0 \coloneqq \emptyset$, $X^{\alpha +1} = \sem{\varphi}{w}{[X \mapsto X^\alpha]}$ and $X^\lambda = \bigcup_{\beta < \lambda} X^\beta$,
for limit ordinals $\lambda$. The least ordinal $\alpha$ such that $X^\alpha = X^{\alpha+1}$ is then the closure ordinal of $\mu X.\ \varphi$ on $\M{T}$,
denoted by $\cl{\varphi}{\M{T}}$. Given a formula, the supremum of $\cl{\varphi}{\M{T}}$ over all transition systems, if it exists, is then the closure
ordinal of the formula. Conversely, given a transition system, the supremum of $\cl{\varphi}{\M{T}}$ over all formulas is then the closure
ordinal of the system.

Of course, over infinite words, the closure ordinal of a formula is always at most $\omega$. We say that a formula $\mu X.\ \varphi$ is equivalent to its $i$th unfolding on $w$ if $\cl{\mu X.\ \varphi}{w}$ is $i$. Moreover, $\mu X.\ \varphi$ is equivalent to a finite unfolding if there is $i$ such that it is equivalent
to its $i$th un folding.
In this case, we also have access to a syntactic unfolding, obtained via
$X^0 = \myfalse$ and $X^{i+1} = \varphi[X^i/X]$, i.e., by syntactically unfolding the fixpoint definition $i$ times, where the remaining occurrences
of $X$ are replaced by $\myfalse$. For greatest fixpoints, an analogous version with $\mytrue$ instead of $\myfalse$ is available.

The situation becomes more involved when dealing with multiple and nested fixpoints, in particular if these are of differing polarity and mutually recursive.
One reason for the added difficulty is the organization of an unfolding process, which, if it is to yield a finite unfolding, needs to terminate. The other 
reason is non-monotonic behavior when dealing with mutually recursive fixpoint of opposite polarities (cf.\ \cite{DBLP:conf/mfcs/BruseSL21}, Ex.~8). However,
if all fixpoint definitions involved are equivalent to some finite unfolding, then the unfolding process is quite robust, and yields equivalent formulas (cf.~\cite{DBLP:conf/mfcs/BruseSL21}, Lemma~7). Hence, we shall speak of \emph{the $i$th unfolding} of a formula $\varphi$ if it is obtained in the following manner:
\begin{definition}{(\cite{DBLP:conf/mfcs/BruseSL21}, Def.~2)}
\label{def:unfolding}
Let $\hat{\mu} = \myfalse$ and $\hat{\nu} = \mytrue$. Let $i \geq 0$ and let $\varphi$ be a formula of the linear-time $\mu$-calculus.
Let $X_0,\dotsc,X_{k-1}$ be an enumeration of the fixpoint variables of $\varphi$ such that $X_i$ is bound by $\sigma X_i.\ \psi_i$, and such that
the $\sigma_i X_i.\ \psi_i$ is not a subformula of $\sigma_g X_j.\ \psi_j$ if $i < j$. Define a sequence $\varphi_{k-1}^n,\dotsc,\varphi_o^n$ with $\varphi_{k-1}^n = \varphi$ and
\[
\psi^0_i = \hat{\sigma_i}, \qquad \qquad \psi_i^{j+1} = \fp_{\varphi^n_i}(X_i)[\psi_i^j/X_i], \qquad \varphi_{i-1}^n = \varphi_i^n[\psi_i^n/\sigma_i X_i.\ \fp_{\varphi_i}(X_i)].
\]
Then $\varphi_0^n$ is the $n$th (bottom-up) unfolding of $\varphi$.
\end{definition}
In other words, this procedure unfolds each fixpoint definition $n$ times, starting with one that is least in the syntax tree, and then continues upwards in the syntax tree. Cf.~\cite{DBLP:conf/mfcs/BruseSL21} for further details.

We then say that, over a word $w$, a formula $\varphi$ is equivalent to its $n$th unfolding $\varphi^n$ if $\varphi^n, \varphi^{n+1},\dotsc$ and $\varphi$
all define the same set on $w$.

\begin{definition}
\label{def:finite-unfolding}
A word $w$ has \emph{finite convergence} of $\mu$-calculus formulas (or simply finite convergence) if every formula $\varphi$ is equivalent to its $n_\varphi$th unfolding over $w$,
where $n_\varphi$ depends only on $w$ and $\varphi$.
\end{definition}
Note that this means that there might be formulas that require arbitrarily large unfoldings. In fact, all but the simplest words behave like this.

\begin{lemma}
\label{lem:closure-trivial}
Let $w$ be a some word. If there is $n$ such that all formulas of the linear-time $\mu$-calculus are equivalent to their $n$th unfolding
(and, hence, the closure ordinal of $w$ is $n$),
then $w$ is ultimately periodic.
\end{lemma}
\begin{proof}
If a word is ultimately periodic then its bisimulation quotient is finite. On a finite structure, any fixpoint definition is equivalent to a finite unfolding.

Conversely, let $w$ be a $\Sigma$-word such that every fixpoint definition is equivalent to its $n$th unfolding. Let $v \in \ifactor{w} \cap \Sigma^n$. By
the pigeonhole principle, $v$ exists. Then the formula $\mu X.\ v \vee \nxt X$, where $v$ is shorthand for $v_0 \wedge \nxt (v_1 \wedge \nxt ( \dots \nxt v_{k}))$,
is equivalent to a finite unfolding, i.e., to $\bigvee_{0 \leq j \leq n} \nxt^j v$. Hence, $v$ appears at most every $n$ positions. Since this holds for any such $v$, we have that $w$ is ultimately periodic of the form $uv^\omega$. 
\end{proof}

A natural follow-up question is for which words all formulas are equivalent to some finite unfolding, but not necessarily with a uniform bound. Obviously,
ultimately periodic words have this property. However, a more interesting question is after words that are not ultimately periodic, but still have finite
convergence. In \cite{DBLP:conf/mfcs/BruseSL21}, a word with this property was presented, and also a sufficient characterization of finite convergence
\begin{proposition}{(\cite{DBLP:conf/mfcs/BruseSL21})}
A word $w$ has finite convergence if all trivial automata have bounded runs on $w$.
\end{proposition}
\label{prop:mfcs}
We refer to \cite{DBLP:conf/mfcs/BruseSL21} for the proof.
The converse is not hard to see either.
\begin{lemma}
\label{lem:conf-tra}
If a word has finite convergence, then all trivial automata have bounded runs on it.
\end{lemma}
\begin{proof}
Let $\M{A} = (\{q_0,\dotsc,q_{k-1}\}, \Sigma, \delta, q_I, F, b)$ be a trivial automaton. Define the following system of fixpoints:
\[
\mu \left( \begin{aligned} X_0.\ & \bigwedge_{a \in \Sigma} a \rightarrow \nxt X_{\delta(q_0, a)} \\ \vdots & \quad \vdots \\ X_{k-1}.\ & \bigwedge_{a \in \Sigma} a \rightarrow \nxt X_{\delta(q_{k-1}, a)} \end{aligned} \right)
\]
Note that this yields the empty conjunction, and, hence, $\mytrue$ for final states.
By the Beki\`c Lemma \cite{Bekic84}, this system can be converted into a standard formula of the linear-time $\mu$-calculus, which then is equivalent
to a finite unfolding, i.e., a formula of basic modal logic. Since truth of such a formula at some position depends only the next $n$ positions,
where $n$ is the modal depth of this formula, and since the trivial automaton has finite runs exactly on positions that are models of this formula,
the runs of $\M{A}$ are bounded by $n$.
\end{proof}

Finally, we observe that the word from \cite{DBLP:conf/mfcs/BruseSL21} is almost periodic. This is not a coincidence.
\begin{lemma}
\label{lem:tra-almost}
Let $w$ be a word on wich all trivial automata have bounded runs. Then $w$ is almost periodic.
\end{lemma}
\begin{proof}
Let $v$ be any factor of $w$, and let $\varphi_v = \mu X.\ v \vee \nxt X$, where again  $v$ is shorthand for $v_0 \wedge \nxt (v_1 \wedge \nxt ( \dots \nxt v_{k}))$.
Let $\M{A}_v$ be a trivial automaton for $\varphi_v$. It can be obtained directly using the Knuth-Morris-Pratt Algorithm \cite{DBLP:journals/siamcomp/KnuthMP77}.
Since all runs of $\M{A}_v$ stop after at most $k$ steps for some $k$, the factor $v$ either does not occur after a finite prefix, or it occurs within
each factor of length $k$. Since $v$ was arbitray, $w$ is almost periodic.
\end{proof}
The rest of the paper is dedicated to showing the converse, i.e., that trivial automata have bounded runs on almost-periodic words. Together with Prop.~\ref{prop:mfcs} and Lemmas~\ref{lem:conf-tra} and \ref{lem:tra-almost}, this yields a complete characterization of the words with finite fixpoint convergence
as exactly those that are almost periodic.

\section{Characterization of Finite Fixpoint Convergence}
\label{sec:finite}
For the remainder of this section, fix an alphabet $\Sigma$. All languages
are $\Sigma$-languages. We now show that almost-periodic words have finite convergence,
by showing that all trivial automata have bounded runs on them. The main
argument consists of rewriting a trivial automaton into an equivalent regular expression,
maintaining properties such as prefix-freeness of suitable constituent sub-espressions
throughout the rewriting process.

\begin{lemma}
\label{lem:fac:finite}
Let $w$ be almost periodic, and let $L$ and $L'$ be two \emph{finite} languages of finite words such
that $L^* L'$ is prefix-free. Then $L^* L \cap \factor{w}$ is finite.
\end{lemma}
\begin{proof}
It suffices to show this for $L' = \{v\}$ for some $v$, since $L^* L' = \bigcup_{v \in L'} L \{v\}$.
Moreover, it is not hard to see that, if $L^* \{v\} \cap \factor{w}$ is infinite, then 
$L^*\{v\} \cap \ifactor{w}$ is also infinite. 

For the sake of contradiction, assume that $L^* \{v\} \cap \factor{w}$, and, hence, $L^*\{v\} \cap \ifactor{w}$ are infinite, for a finite,
prefix-free $L$. We can assume that $v \notin L$. Let $k = \max_{ v' \in L \cup \{v\}} |v'|$. We define a sequence $v_0,\dotsc,v_k$
as follows: $v_0 = \epsilon$, and for $i < k$, let $v_{i+1}$ be such that (\textsc{i}) $|v_{i+1}| \geq \recurrence{v_iv}{}$
and (\textsc{ii}) $v_{i+1} \in L^*$ and (\textsc{iii}) $v_{i+1} v \in \ifactor{w}$. This sequence is well-defined since, by assumption $\ifactor{w} \cap L^* \{v\}$ is infinite, whence $\recurrence{v_i v}{}$ is well-defined.

We now have the following: for each $0 \leq i \leq k$, we can decompose $v_i v$ into
$u_0^i \dotsb i_{m_i}^i$ where $u_{j}^i \in L$ if $j < m_i$ and $u_{m_i}^i = v$, i.e.,
$v_i v$ consists of blocks of length at most $k$ from $L \cup \{v\}$.  Moreover,
since $|v_{i+1}| \geq \recurrence{v_i v}{}$, for each $0 < i \leq k$ there is $p_i < |v_{i}| - | v_{i-1}v|$
such that $v_i = u v_{i-i}v u'$ for $u \in \Sigma^{p_i}$, i.e., $v_{i-1}v$ appears as an infix
of $v_i$ at position $p_i$. If there are multiple occurrences of $v_{i-1}v$, take the first one
to define $p_i$. Let $P_i$ be defined as follows:
\[
P_i = \bigcup_{0 \leq j \leq m_i} \{\big(\sum_{l = k-1}^i p_i\big) + \sum_{l=0}^{j} |u_l^i| \}
\] 
Hence, $P_i$ collects the starting positions of the individual blocks of $v_i v$
inside $v_k$. Note that $P_0$ is a singleton, as it consists solely of $p = \sum_{l = k-1}^0 p_i$.
Consider the interval of positions $p-k+1,\dotsc,p$. It is of length $k$, whence
it must contain an element of $P_i$ for all $0 \leq i \leq k$. Due to the pigeonhole principle, 
there is a position $p'$ with $p_1+k \leq p' \leq p$ such that $p' \in P_j$ and $p' \in P_{j'}$ for
$j < j'$. Hence, there is $m$ such that $u_m^j \dotsb u_{m_j}^{j}$ starts at $p$ and,
there is $m'$ such that $u_0^{j'}\dotsb u_{m'}^{j'}$ starts at $\sum_{l = 0}^{j'} p_l$
and ends at $p$. But then $v' = u_0^{j'}\dotsb u_{m'}^{j'} u_{m}^{j}\dotsb u_{m_j}^j \in L^* \{v\}$
and $v'$ is also a strict prefix of $v_{j'} v$. This contradicts prefix-freeness of $L^* \{v\}$, which
proves the lemma.
\end{proof}

\subsection{Conversion of \texorpdfstring{NFA$_\text{reg}$}{NFAreg} into Regular Expressions}

\begin{definition}
\label{def:ajoin}
Let $\M{A} = (Q, \Sigma, \delta, I, F)$ and $\M{A}' = (Q, \Sigma, \delta', I, F)$ be NFA$_\text{reg}$ over the
same alphabet with the same sets of states, initial states and final states. Then 
$\Transition{\M{A}}{\ajoin{p}{q}{r_1}{r_2}}{\M{A}'}$ if the following are true:
\begin{itemize}
\item $p, q \in Q$
\item $(p, r_1, q), (p, r_2, q) \in \delta$
\item $\delta' = (\delta \setminus \{(p, r_1, q), (p, r_2, q) \}) \cup \{(p, r_1 + r_2, q)\}$.
\end{itemize}
\end{definition}
Intuitively, this means that $\M{A}'$ is obtained by replacing two parallel edges in $\M{A}$, labeled
by $r_1$ and $r_2$, by an edge $r_1 + r_2$.

\begin{definition}
\label{def:remove}
Let $\M{A} = (Q, \Sigma, \delta, I, F)$ and $\M{A}' = (Q', \Sigma, \delta', I, F)$ be NFA$_\text{reg}$ over the
same alphabet with the same sets of states, initial states and final states. Then 
$\Transition{\M{A}}{\aremove{q}}{\M{A}'}$ if the following are true:
\begin{itemize}
\item for all $p, p' \in Q$, there is at most one $r$ with $(p, r, p') \in \delta$,
\item $q \in Q \setminus (I \cup F)$,
\item $Q' = Q \setminus \{q\}$,
\item $\delta' = \delta'' \cup \{(p, r_1 r_2^* r_3, p') \mid (p, r_1, q) \in \delta, (q, r_2, q) \in \delta, (q, r_3, p') \in \delta\}$ where $\delta'' = \delta \setminus \{(p, r, p') \mid p = q \text{ or }p' = q\}$.
\end{itemize}
\end{definition}
This means that $\M{A}'$ is obtained by removing the non-initial, non-final state $q$ and condensing all paths that go from $p$ directly to $q$ using the unique edge labeled $r_1$, potentially loop there via the loop labeled $r_2$, and then go directly to $p'$ via the unique edge labeled $r_3$, into a single edge from
$p$ to $p'$, labeled by $r_1 r_2^* r_3$. Note that the first condition requires that no parallel edges exist in $\M{A}$.

We write $\Transition{\M{A}}{}{\M{A}'}$ if either we have  $\Transition{\M{A}}{\ajoin{p}{q}{r_1}{r_2}}{\M{A}'}$
for some $p, q, r_1, r_2$, or if we have $\Transition{\M{A}}{\aremove{q}}{\M{A}'}$ for some $q$.
The following is folklore, see also \cite{esparza2023automata}.
\begin{lemma}
\label{lem:conversion:folklore}
If $\Transition{\M{A}}{}{\M{A}'}$ then $L(\M{A}) = L(\M{A}')$. In fact, for all states $p$ of $\M{A}'$,
we have $L(\M{A}'_p) = L(\M{A}_p)$.

Moreover, for each $\M{A}$ there is a sequence $\M{A} = \M{A}_0,\dotsc,\M{A}_{k-1}$
such that $\Transition{\M{A}_i}{}{\M{A}_{i+1}}$ for all $0 \leq i < k-1$ and
such that $\M{A}_{k-1}$ has only initial and final states.  
\end{lemma}

\subsection{Properties of the Conversion Process}

Given an NFA$_\text{reg}$ $(Q, \Sigma, \delta, I, F)$, define $\aout{q} = \{r \mid (q, r, q') \in \delta\}$
and $\laout{q} = \bigcup_{r \in \aout{q}} L(r)$.

\begin{definition}
Let $\M{A} = (Q, \Sigma, \delta, I, F)$ be an NFA$_\text{reg}$. Then $\M{A}$ is \emph{minimal} if it satisfies the 
following:
\begin{itemize}
\item $I = \{q_I\}$, 
\item $F = \{q_F\}$ and there are no $r, q$ with $(q_F, r, q) \in \delta$,
\item $L(\M{A}_q) \not = L(\M{A}_{q'})$ if $q \not = q'$.
\end{itemize}
Moreover, $\M{A}$ is \emph{good} if it is minimal and it further satisfies the following:
\begin{enumerate}
\item For all $q \in Q$, the language $\laout{q}$ is prefix-free.
\item $\epsilon \notin \laout{q}$ for all states $q$.
\item For all $q \in Q$, all $r_1, r_2$ in $\aout{q}$, we have that $L(r_1) \cap L(r_2) = \emptyset$.
\end{enumerate}
Finally, let $w$ be an almost-periodic, infinite word. Then $\M{A}$ is \emph{simple} w.r.t.\ $w$ if, for all $q, q' \in Q$
and all $(q, r, q')$, we have that $L(r) \cap \factor{w}$ is finite.
\end{definition}
Intuitively, minimality requires $\M{A}$ to have unique initial and final states, and to be minimized in the classical sense.
Goodness is a sanity criterion that is trivially satisfied by e.g., deterministic automata, and simplicity is required 
to apply Lemma~\ref{lem:fac:finite}.

\begin{lemma}
\label{lem:keep-minimal}
Let $\M{A}$ be an NFA$_\text{reg}$ and let $\Transition{\M{A}}{}{\M{A}'}$. If it is minimal
then $\M{A}'$ is so. 
\end{lemma}
\begin{proof}
It is relatively easy to verify that this holds if $\Transition{\M{A}}{\ajoin{p}{q}{r_1}{r_2}}{\M{A}'}$,
since this action does not change $\laout{q}$ for any state $q$, and it changes $\aout{q}$ only in a harmless way.
Hence, assume that we have that x$\Transition{\M{A}}{\aremove{q}}{\M{A}'}$.
Let $\M{A} = (Q, \Sigma, \delta, \{q_I\}, \{q_F\})$. Since $q \not = q_I$ and $q \not = q_F$, let $\M{A}' = (Q', \Sigma, \delta', \{q_I\}, \{q_F\})$ where $Q' = Q \setminus \{q\}$. Moreover,
for all $p, p'$ there is at most one $r$ such that $(p, r, p') \in \delta$. Hence, $\delta' = \delta'' \cup \{(p, r_1 r_2^* r_3, p') \mid (p, r_1, q) \in \delta, (q, r_3, p') \in \delta, (q, r_2, q) \in \delta\}$
where $\delta'' = \delta \setminus \{(p, r, p') \mid p = q \text{ or } p' = q\}$.

It remains to show that $L(\M{A}'_p) \not = L(\M{A}'_{p'})$ if $p \not = p'$.
To obtain this, observe that $L(\M{A}_p) = L(\M{A}'_p)$ for all $p \in Q'$, which holds due to 
Lemma~\ref{lem:conversion:folklore}.   
\end{proof}

\begin{lemma}
\label{lem:keep-goodness}
Let $\M{A}$ be an NFA$_\text{reg}$ and let $\Transition{\M{A}}{}{\M{A}'}$. If $\M{A}$ is good 
(and therefore minimal), then $\M{A}'$ is so.
\end{lemma}
\begin{proof}
By Lemma~\ref{lem:keep-minimal}, $\M{A}'$ is minimal. Again, the case of 
$\Transition{\M{A}}{\ajoin{p}{q}{r_1}{r_2}}{\M{A}'}$ is straightforward for goodness, so assume
that  $\Transition{\M{A}}{\aremove{q}}{\M{A}'}$ for some $q$.

The second item of goodness follows from goodness of $\M{A}$. Since $\epsilon \notin
\laout{p}$ for all $p$, in particular $\epsilon \notin r$ for $(p, r, p') \in \delta$ unless
$p = q_I$. Hence, $\epsilon \notin r$ for $(p, r, p') \in \delta''$, whence
$\epsilon \notin L(r_1 r_2^* r_3)$ for any $(p, r_1r_2^* r_3, p') \in \delta'$ since $\epsilon \notin r_3$.

Towards the first and third item, let $p \in Q'$ and let $r'_1, r'_2$ be such that $(p, r'_1, p_1) \in \delta'$ and $(p, r'_2, p_2) \in \delta'$ for some, not necessarily distinct $p_1, p_2$. If these edges are also in $\delta$, we are done. Hence, assume
that $r'_1 = r_1 r_2^* r_3$ for $(p, r_1, q), (q,r_2,q), (q, r_3, p_1) \in \delta$. There are two subcases:
\begin{itemize}
\item  If $r'_2 = r_1 r_2^* r'_3$ with $(q, r'_3, p_2) \in \delta$, let $u_1 \in L(r'_1)$
and $u_2 \in L(r'_2)$ be arbitrary. We show that they are incomparable, i.e., neither is a prefix of the other, unless $p_2 = p_1$ (and, hence, $r_3 = r'_3$),
in which case they can also be equal. 
We have that $u_1 = u_1^1 u_1^2 u_1^3$ 
such that $u_1^1 \in L(r_1)$, $u_2^1 \in L(r_2^*)$ and $u_3^1 \in L(r_3)$. Moreover, $u_2 = u_1^2 u_2^2 u_3^2$
with $u_1^2 \in L(r_1)$, $u_2^2 \in L(r_2^*)$ and $u_3^2 \in L(r'_3)$. 

If $u_1^1$ and $u_1^2$ are not incomparable, it follows that $u_1^1 = u_1^2$ since
otherwise $L(r_1)$ and, hence $\laout{p}$, would not be prefix-free, contradicting the assumption that $\M{A}$ is good.

If $u_1^1 = u_1^2$, then also $u_2^1 = u_2^2$ if they are not incomparable. To see this, assume that e.g., $u_2^1$ is a prefix of $u_2^2$. Then either $u_3^1$ is a prefix of the remainder of $u_2^2$, or the other way around. In either case, $\laout{q}$ would not be prefix-free. The other case is symmetric.

Finally, $u_3^1$ and $u_3^2$ must be incomparable if $p_1 \not = p_2$, since otherwise $L(r_3) \cup L(r'_3) \subseteq \laout{q}$ would not be prefix-free or $L(r_3)$ and $L(r'_3)$ would not be disjoint, contradicting
the first or third item of goodness of $\M{A}$. If $p_1 = p_2$, then by the same argument, $u_3^1$ and $u_3^2$ are
either incomparable or equal, for otherwise $L(r_3)$ would not be prefix-free.

\item The other subcase is that there is $r'_2$ with $(p, r'_2, p'')$ such that there are $u_1 \in L(r'_1)$
and $u_2 \in L(r'_2)$. 
As before, we have that $u_1 = u_1^1 u_1^2 u_1^3$ 
such that $u_1^1 \in L(r_1)$, $u_2^1 \in L(r_2^*)$ and $u_3^1 \in L(r_3)$. Then, $u_1^1$ and $u_2$ are
 incomparable, for otherwise $L(r_1) \cup L(r'_2) \subseteq \laout{p}$ would not be prefix-free or $L(r_1)$ 
and $L(r'_2)$ would not be disjoint.
\end{itemize}
 
\end{proof}

\begin{lemma}
\label{lem:keep-simple}
Let $\M{A}$ be an NFA$_\text{reg}$ and let $\Transition{\M{A}}{}{\M{A}'}$. If it is good and simple w.r.t.\ some almost-periodic $w$, then $\M{A}'$ is so. 
\end{lemma}
\begin{proof}
By Lemma~\ref{lem:keep-goodness}, $\M{A}'$ is good. Again, the case of 
$\Transition{\M{A}}{\ajoin{p}{q}{r_1}{r_2}}{\M{A}'}$ is straightforward, so assume
that  $\Transition{\M{A}}{\aremove{q}}{\M{A}'}$ for some $q$.

It suffices to show that $L(r_1 r_2^* r_3) \cap \factor{w}$ is finite for a new edge in $\delta' \setminus \delta$,
since this is already the case for the edges that are also present in $\delta$. By the induction
hypothesis, all of $L(r_1) \cap \factor{w}$, $L(r_2) \cap \factor{w}$ and $L(r_3) \cap \factor{w}$ are
finite.
Since $\M{A}$ is good, these are also all prefix-free. Hence
$L(r_2^* r_3)$ is also prefix-free and with it $(L(r_2) \cap \factor{w})^* (L(r_3) \cap \factor{w})$.
Write $L_2$ for $L(r_2) \cap \factor{w}$ and $L_3$ for $L(r_3) \cap \factor{w}$. 
By Lemma~\ref{lem:fac:finite}, $(L_2^* L_3 ) \cap \factor{w}  = (L(r_2^*) L(r_3)) \cap \factor{w}$ is also
finite. But then $L(r_1 r_2^* r_3) \cap \factor{w}$ is also finite, which is as desired.
\end{proof}

\subsection{The Result}

\begin{theorem}
\label{thm:finite-convergence}
Trivial Automata have bounded runs on almost-periodic words.
\end{theorem}
\begin{proof}
Let $w$ be an almost-periodic word and
let $\M{A} = (Q, \Sigma, \delta, q_I, F, b)$ be a trivial automaton. Note that $\delta$ is deterministic on
$Q \setminus F$.  W.l.o.g.\ $L(\M{A}_q) \not = L(\M{A}_{q'})$
for $q \not = q$ (otherwise, run the usual DFA minimization algorithm). Note that $q_I$ cannot be
a final state unless $L(\M{A}) = \{\epsilon\}$, in which case the result is immediate. Hence, also $F = \{q_F\}$ whence, $\M{A}$ is minimal. 

Moreover, since $\M{A}$ is a DFA, we have that if $(q, r, p) \in \delta$, then $r = a$ for some $a \in \Sigma$,
and for each $q$ and $a$ there is at most one $p$ such that $(q, a, p) \in \delta$. Hence, $\M{A}$ is good, and
it is simple trivially since $\{a\} \cap \factor{w} = \{a\}$ or $\{a\} \cap \factor{w} = \emptyset$ for all
$a \in \Sigma$.

Hence, we can run the conversion algorithm mentioned in Lemma~\ref{lem:conversion:folklore} to convert $\M{A}$ into an equivalent $\M{A}'$ that has only initial and final states. Since $\M{A}$ has exactly one of each,
this means that $\M{A}' = (\{q_I, q_F\}, \Sigma, \delta, q_I, \{q_F\}, b)$. It is easy to see that $q_{F}$ still
has no outgoing edges, and that $\delta = \{(q_I, r, q_I), (q_I, s, q_F)\}$ for some $r, s$. Hence
$L(\M{A}') = L(r^* s)$. 

By Lemmas~\ref{lem:keep-minimal}, \ref{lem:keep-goodness} and \ref{lem:keep-simple}, $\M{A}'$ is also minimal, good and simple. 
Hence, both $L(r) \cap \factor{w}$ and $L(s) \cap \factor{w}$ are finite, and $L(r) \cap L(s) = \emptyset$, and, moreover, $\laout{q_I} = L(r) \cup L(s)$ is prefix-free. It is now not hard to see that $L(r^* s) \cap \factor{w}$
is also finite, using Lemma~\ref{lem:fac:finite} and the fact that $L(r^*s) \cap \factor{w} \subseteq (L(r)^* \cap \factor{w}) \cdot L(s) \cap \factor{W}$, the latter component of which was already seen to be finite.

Hence, $L(\M{A}) \cap \factor{w}$ is also finite, since it agrees with $L(r^* s) \cap \factor{w}$. Hence there are only finitely many words that are a factor of $w$ and for which there is a run on $\M{A}$ from
$q_I$ to some final state. Let $k$ be the length of the longest such word. Then all runs of $\M{A}$ on $w$ either stop after at most $k$ steps, or go on forever.
\end{proof}

This also yields a decidability result for $\mucalc$ on almost-periodic words.
\begin{proposition}
Let $w$ be an effectively almost-periodic word, and assume that, for each $v \in \factor{w}$, the recurrence distance $\recurrence{v}{w}$ is computable.
Then $\mucalc$ is decidable on $w$.
\end{proposition}
We omit a proof, since the result is already known \cite{semenov}. A proof consists of the observation that, from
\cite{DBLP:conf/mfcs/BruseSL21}, we can rewrite $\varphi$ into a trivial automaton $\M{A}_\varphi$. The bound on the runs 
of $\M{A}_\varphi$ can be computed by making the proofs above effective. This is rather tedious and yields a non-elementary
upper bound (similar to \cite{semenov}.

 We close this section with the following remark: Since the $\omega$-regular languages are exactly those defined by the linear-time $\mu$-calculus,
any $\omega$-regular language is equivalent to some trivial automaton. But we have just seen that such a trivial automaton can also be rewritten into
a finite-word regular expression, which captures exactly the prefix-free finite-word language of factors on which the trivial automaton stops. For the sake
of simplicity, we assume for the moment that this means that it accepts.
\begin{remark}
\label{rem:nf}
Any $\omega$-regular expression of the form $\bigcup_{i=0}^{n-1} r_i s_i^\omega$ can be rewritten into an expression of the form $r \Sigma^\omega$ such that the two agree on the set of almost-periodic words.
\end{remark}
Note that the latter is an $\omega$-regular expression that only has one component, but also a trivial infinitary part.

\section{Conclusion}
\label{sec:concl}
We have settled the open question from \cite{DBLP:conf/mfcs/BruseSL21} after the class
of words on which all fixpoints are equivalent to some finite approximation, by showing
that this is exactly the class of almost-periodic words. Besides re-obtaining a known decidability
result for $\mucalc$ on these words, this also yields a normal form for $\omega$-regular expressions
over almost-periodic words, and a simple way to compute them.

Directions for further research comprise extending the results to more complicated classes of structures.
The case of bi-infinite words is already under investigation\footnote{Personal communication from Damian Niwinski.},
and a natural follow-up question concerns infinite trees. Here, the class of trees under investigation
is important, since automata-theoretic methods on trees rely on the choice of tree (ranked or unranked, ordered
or unordered). While these classes can be translated into each other, some of these translations 
do not keep closure ordinals: on binary trees, due to continuity reasons, the closure ordinal of
a formula is always at most $\omega$, while it is straightforward to generate trees with large countable closure
ordinals if branching can be infinite.

 \bibliographystyle{plain}
 \bibliography{literature}

\end{document}